\documentclass[11pt,reqno]{amsart}
\usepackage[utf8]{inputenc}
\usepackage[english]{babel}
\usepackage{csquotes}

\usepackage[
	backend=biber,
	style=alphabetic,
	sorting=nyt,
	url=false,
	eprint=true
]{biblatex}
\usepackage{hyperref}
\AtBeginBibliography{\small}

\usepackage{geometry}
\usepackage{amsmath, amssymb, amsthm, mathtools, cancel, graphicx, wrapfig, forest, verbatim, amsfonts, bbm, caption, subcaption} 

\usepackage{pgfplots}
\pgfplotsset{compat=1.18}

\usepackage{marvosym}

\newtheorem{thm}{Theorem}[section]
\newtheorem{cor}[thm]{Corollary}

\newtheorem{lem}[thm]{Lemma}
\newtheorem{prop}[thm]{Proposition}

\theoremstyle{definition}
\newtheorem{defn}[thm]{Definition}
\newtheorem{eg}[thm]{Example}

\theoremstyle{remark}
\newtheorem*{rem}{Remark}

\numberwithin{equation}{section}

\usepackage{parskip}
\usepackage{setspace}
\usepackage{mathrsfs}

\makeatletter
\def\thm@space@setup{%
  \thm@preskip=\parskip \thm@postskip=0pt
}
\makeatother

\usepackage{enumitem}

\newcommand{\N}{\mathbb{N}}
\newcommand{\Z}{\mathbb{Z}}
\newcommand{\Q}{\mathbb{Q}}
\newcommand{\R}{\mathbb{R}}

\renewcommand{\S}{\mathbb{S}}

\newcommand{\T}{\mathbb{T}}

\newcommand{\ep}{\varepsilon}
\newcommand{\ph}{\varphi}
\newcommand{\longto}{\longrightarrow}

\newcommand{\id}{\mathrm{id}}

\newcommand{\Hmm}[1]{\leavevmode{\marginpar{\tiny%
$\hbox to 0mm{\hspace*{-0.5mm}$\leftarrow$\hss}%
\vcenter{\vrule depth 0.1mm height 0.1mm width \the\marginparwidth}%
\hbox to
0mm{\hss$\rightarrow$\hspace*{-0.5mm}}$\\\relax\raggedright #1}}}

\begin{document}
\title{Multifractality of Semiclassical Measures on Star Graphs}
\author[M.~Nietschmann]{Marius Nietschmann}

\address{M.~Nietschmann, Sorbonne Université and Université Paris Cité, CNRS, IMJ-PRG, F-75005 Paris, France}

\email{marius.nietschmann@imj-prg.fr}
\date{\today}

\begin{abstract}
	We study eigenfunctions of quantum star graphs in the large edge number limit through the edge-mass distributions associated with their semiclassical measures. For generic edge lengths, we show that these distributions can realize every admissible multifractal scaling law along suitable subsequences of eigenvalues. We also prove a constructive result for quasi-equilateral star graphs. Starting from prescribed probability measures, we construct graphs and locate eigenvalues inside spectral clusters whose eigenfunctions reproduce the same scaling behavior. These results show that quantum star graphs form an explicit model realizing the full range of admissible multifractal behavior between localization and equidistribution.	
\end{abstract}

\maketitle

\section*{Setting the Stage}

The spatial distribution of eigenfunctions is one of the central questions in quantum chaos. A guiding principle is that spectral and eigenfunction statistics should reflect the nature of the underlying classical dynamics. For generic integrable systems, the Berry--Tabor conjecture predicts Poissonian spectral statistics, while for chaotic systems the Bohigas--Giannoni--Schmit conjecture predicts the universal statistics of random matrix theory~\cite{BT77,BGS84}. On manifolds with ergodic geodesic flow, a density-one subsequence of eigenfunctions becomes equidistributed~\cite{S74,CdV85,Z87}. In negatively curved or Anosov manifolds, this delocalization is known to satisfy stronger restrictions, such as entropy lower bounds for semiclassical measures~\cite{AN07,A08} and lower bounds on the mass of eigenfunctions on open sets~\cite{DJ18,DJN22}.

At the opposite end of the spectrum lies localization. For instance, the round sphere admits spherical harmonics that concentrate on the equator and saturate Sogge's $L^p$-bounds~\cite{S88}. The dichotomy between localization and delocalization also appears in more complex systems as the Anderson model, introduced to describe absence of diffusion in disordered media~\cite{A58}. In the localized regime, a rich rigorous theory proves pure point spectrum, exponentially localized eigenfunctions, dynamical localization, and Poissonian spectral statistics~\cite{FS83,FMSS85,M96,GK14,AW15}. By contrast, the delocalized regime of the Anderson model, and especially the transition between localized and extended states, remains much less understood. The physics prediction is that at the mobility edge eigenfunctions are neither localized on a small region nor uniformly spread over the whole system. Instead, they are expected to exhibit multifractal behavior: their mass is distributed over many sites, but in a highly inhomogeneous way governed by nontrivial scaling laws~\cite{CP86,SG91,KLAA94,EM08}.

A second class of intermediate systems are pseudo-integrable systems. These are neither completely integrable nor uniformly chaotic. Classical examples include rational polygonal billiards and related singular billiards~\cite{RB81}. Their spectral statistics are expected to be intermediate between Poisson and random-matrix statistics, and their eigenfunctions are natural candidates for multifractal behavior~\cite{BGS99}. Rigorous results in this direction are difficult, since such systems are too singular for many standard tools of quantum chaos and too structured for probabilistic methods. In recent works, Keating and Ueberschär constructed multifractal eigenfunctions in singular billiards~\cite{KU22} and in quantum star graphs~\cite{KU}, providing a rigorous model for this intermediate regime. The present paper continues and substantially generalizes the results in~\cite{KU}.

We focus on quantum star graphs because they provide an explicit yet nontrivial model of intermediate behavior. Quantum graphs are well-established models in quantum chaos and wave propagation~\cite{KS99,GS06,BK13}. Star graphs, in particular, have long been used as tractable models for spectral and eigenfunction statistics beyond the purely integrable or chaotic paradigms. Their two-point correlations are related to those of \v{S}eba billiards~\cite{BK99,BBK01}, their eigenfunction value distributions can be analyzed explicitly~\cite{KMW03}, and they are known not to satisfy quantum ergodicity in the large-graph limit~\cite{BKW04}. Thus, although star graphs are much simpler than the Anderson model or polygonal billiards, they retain a key feature relevant to the present question: their eigenfunctions are explicit enough to study mass distribution directly, while still allowing behavior which is neither localized nor equidistributed.

The limit considered here is the large edge number limit. This is natural for the study of multifractality, since multifractality is a scaling phenomenon and therefore requires a system size that grows compared to the wavelength. On a fixed compact graph, one may study semiclassical measures, but there is no increasing spatial scale on which to measure nontrivial distribution laws. Letting the number of edges grow provides such a scale. In this limit, the edge-mass distribution of an eigenfunction becomes a probability measure on a finite set whose size tends to infinity. This is the setting in which localization, equidistribution, and intermediate multifractal behavior can be compared directly.

Throughout the paper, we study scaling laws through mass exponents. They are abstractly defined for a sequence of probability measures $\mu_n$ on $\{1,\dots,n\}$ via
	\[ \tau(q) = \limsup_{n\to\infty} \frac{\log\sum_i \mu_n(i)^q}{\log n}, 
		\qquad q\geq1. \] 
In the application to star graphs, the measures will be given by the amplitude squared of eigenfunctions on each edge. For eigenfunctions that are perfectly equidistributed or supported on a bounded number of edges, one obtains an affine mass exponent. Accordingly, we speak of multifractality when $\tau$ is non-affine, reflecting the nontrivial scaling described above. 

The results of this paper show that quantum star graphs are completely flexible at the level of these discrete scaling laws. We first characterize the possible scaling laws for probability measures on growing finite sets. We then show that every such law is realized by eigenfunctions of generic quantum star graphs along a suitable subsequence as the number of edges tends to infinity. Since this subsequence is obtained through an ergodicity argument and the corresponding eigenvalues are not explicitly located, we also study quasi-equilateral star graphs. Their spectrum forms clusters, and we show how to construct graphs and choose eigenvalues inside prescribed clusters so that the resulting eigenfunctions reproduce any admissible scaling behavior. Thus, the examples of Keating and Ueberschär are part of a larger picture: quantum star graphs can realize the full range of admissible multifractal behavior.

\subsection*{Acknowledgments}

The author would like to thank his supervisors, Maxime Ingremeau and Henrik Ueberschär, for their guidance and advice throughout the course of this work.

\section{Set-Up and Main Results} \label{sec: set-up}

\subsection{Star Graphs}

Let $X^n$ be a (combinatorial) star graph of $n$ edges. Whenever the number of edges will not be of importance in the discussion, we will drop the index~$n$. Together with a choice $L=(L_1,\dots,L_n)\in\R_+^n$, where $\R_+=(0,\infty)$, the pair $(X,L)$ defines a \emph{metric graph} obtained by gluing together $n$ intervals of lengths given by $L$ at the central vertex that we denote by $v$ using the combinatorics of $X$ (confer Figure~\ref{fig: star graph}). Each exterior vertex of $X$ is identified with $0$ of each interval $[0,L_i]$. We can equip the metric graph $(X,L)$ with the Laplace operator $\Delta_L$ acting on functions $\ph=(\ph_i)_{1\leq i\leq n}$ in the Sobolev space $H^2(X,L) = \bigoplus_i H^2([0,L_i])$ satisfying the following conditions: 
\begin{enumerate}[label = \emph{\roman*})]
	\item $\ph$ is continuous at the central vertex: 
	\begin{equation}\label{eq: ph cont}
		\forall i,j\in\{1,\dots,n\} : \ph_i(L_i) = \ph_j(L_j). 
	\end{equation}
	\item $\ph$ vanishes on the exterior vertices (Dirichlet boundary): 
	\begin{equation}\label{eq: ph Dirichlet}
		\forall i\in\{1,\dots,n\} : \ph_i(0) = 0. 
	\end{equation}
	\item $\ph$ satisfies the $\delta$-type condition: 
	\begin{equation}\label{eq: ph delta type}
		\sum_{i=1}^n \ph_i'(L_i) = \alpha \, \ph(v), \qquad \alpha\in\R. 
	\end{equation}
\end{enumerate}
The Laplacian is then given by $-d^2/dx^2$ on each edge. A metric graph $(X,L)$, together with the previously defined $\Delta_L$, is called a \emph{quantum star graph}. The conditions~(\ref{eq: ph cont})-(\ref{eq: ph delta type}) guarantee that the Laplacian is a self-adjoint operator. 

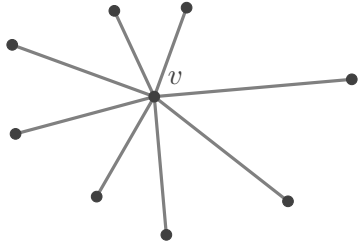
\begin{figure}[h]
\begin{tikzpicture}[scale=.5]
	\foreach \angle in {-20,15,62,105}{
		\draw[gray,very thick] (0,0) -- (\angle-100:.5*\angle+3.4cm);
		\draw[darkgray,fill] (\angle-100:.5*\angle+3.4cm) circle (4pt);
	}
	\draw[gray,very thick] (0,0) -- (70:2.5cm);
	\draw[darkgray,fill] (70:2.5cm) circle (4pt);
	\draw[gray,very thick] (0,0) -- (195:3.8cm);
	\draw[darkgray,fill] (195:3.8cm) circle (4pt);
	\draw[gray,very thick] (0,0) -- (115:2.5cm);
	\draw[darkgray,fill] (115:2.5cm) circle (4pt);
	\draw[gray,very thick] (0,0) -- (160:4cm);
	\draw[darkgray,fill] (160:4cm) circle (4pt);
	\draw[darkgray,fill] (0,0) circle (4pt) node[label={[label distance=-3pt]30:$v$}]{};
\end{tikzpicture}
\caption{A quantum star graph with $n=8$ edges.}
\label{fig: star graph}
\end{figure}

\subsection{Spectrum and Eigenfunctions}

Since $\Delta_L$ has compact resolvent, its spectrum is discrete and accumulates at infinity~\cite[Chapter 3]{BK13}. There exists an orthogonal basis of $L^2(X,L)$ consisting of eigenfunctions of $\Delta_L$ that we will denote $\ph^{\langle k \rangle}$ with associated eigenvalue $k^2$. Since on every edge the restriction $\ph^{\langle k\rangle}_i$ is an eigenfunction on an interval, there exist constants $\alpha_i,\beta_i$ such that $\ph^{\langle k\rangle}_i(x) = \alpha_i \sin kx + \beta_i \cos kx$. The Dirichlet condition~(\ref{eq: ph Dirichlet}) yields $\beta_i=0$. Hence,
\begin{equation} \label{eq: ph is sin}
	\ph^{\langle k\rangle}_i(x) = \alpha_i \sin kx. 
\end{equation}
We will see in Section~\ref{sec: spectral determinant} that if $L$ has rationally independent entries, then $\ph^{\langle k\rangle}$ does not vanish at the central vertex. Since star graphs are trees, this implies that $k^2\in\sigma(\Delta_L)$ is a simple eigenvalue \cite[Corollary 3.1.9]{BK13} and hence $\ph^{\langle k\rangle}$ is unique (up to a multiplicative constant). 

\subsection{Semiclassical Measures and Mass Measures} \label{subsec: mass measures}

A function $\ph$ in $L^2(X,L)$ induces a measure on the quantum graph via $|\ph|^2 \mathrm{vol}_L$. Every subsequence of eigenfunctions $\ph^{\langle k_m\rangle}$ such that the limit 
\begin{equation} 
	\lim_{m\to\infty} |\ph^{\langle k_m\rangle}|^2 \mathrm{vol}_L (f) 
		= \lim_{m\to\infty} 
			\int_X f\,|\ph^{\langle k_m\rangle}|^2 d\mathrm{vol}_L 
		=: \mu_\mathrm{sc}(f)
\end{equation}
exists for all continuous functions $f$ induces a measure $\mu_\mathrm{sc}$ on $(X,L)$ which is called \emph{semiclassical measure}. In other words, semiclassical measures are accumulation points of the sequence $|\ph^{\langle k\rangle}|^2 \mathrm{vol}_L$ of measures obtained from the eigenfunctions where the space of measures on $(X,L)$ is endowed with the weak*-topology. 

We furthermore introduce the discrete measures 
\begin{equation} \label{eq: mass measure}
	\mu_n^k (i) 
		:= \|\ph_i^{\langle k\rangle}\|_\infty^2 
		= |\alpha_i|^2  
\end{equation}
associated with an eigenfunction $\ph^{\langle k\rangle}$ that we will refer to as \emph{mass measures}. We will normalize eigenfunctions in such a way that mass measures are probability measures on $\{1,\dots,n\}$. 

Mass measures and semiclassical measures are closely related. Indeed, we will see in Section~\ref{sec: semiclassical measures} Lemma~\ref{lem:class sc measures} that if the edge lengths are rationally independent, the weak*-limit of $|\ph^{\langle k_m\rangle}|^2 \mathrm{vol}_L$ exists if and only if $\mu_n^{k_m}$ converges in the space of probability measures on $\{1,\dots,n\}$. We then have 
\begin{equation} \label{eq: relation mass and semiclassical measures}
	\mu_\mathrm{sc} 
		= \frac12 \sum_{i=1}^n \left(\lim_{m\to\infty} \mu_n^{k_m}(i)\right) \mathrm{Leb}_i. 
\end{equation}
Moreover, due to~\cite[Theorem 2.1]{CdV15}, every mass measure $\mu_n^k$ is a semiclassical measure on $(X,L)$. 

\subsection{Multifractality}

Our interest lies in studying the behavior of sequences of mass measures as $n$ tends to infinity. To that end, we will introduce the mass exponent that captures phenomena of localization and delocalization of the measures in question. We will identify the set probability measures on $\{1,\dots,n\}$ with 
\begin{equation} \label{eq: identification set of measures and conv}
	\mathrm{conv}\{e_1,\dots,e_n\}\subset\R^n 
\end{equation}
and write $\|\cdot\|_q$ for the $q$-norm. 

\begin{defn}[Mass exponent] \label{defn: mass exponent}
	Let $\mu_n$ be a sequence of probability measures on $\{1,\dots,n\}$. Its \emph{mass exponent} is defined as 
	\begin{equation}
		\tau(q) = \limsup_{n\to\infty} \frac{\log\|\mu_n\|_q^q}{\log n}, 
				\qquad q\geq1. 
	\end{equation}
	We denote by $\mathcal E$ the set of functions $\tau$ that are mass exponents of some sequence of discrete probability measures. 
\end{defn}

We remark that while the quotient $\log(\|\mu_n\|_q^q)/\log n$, in general, does not converge as $n\to\infty$, the limit superior always exists (because $\|\mu_n\|_q^q\leq1$). Therefore, any sequence of measures has a well-defined mass exponent. 

It is an easy computation that a point measure has mass exponent~$0$ and the uniform measure has mass exponent equal to~$1-\id$. The next easiest generalizing example that comes to mind might be uniform measures on subsets of $\{1,\dots,n\}$. Those give rise to affine functions as mass exponents. It is a natural question whether mass exponents can take on a more general form. This inspires the following definition. 

\begin{defn}[Multifractality] \label{defn: multifractality}
	A sequence of probability measures is called \emph{multifractal} if its mass exponent is a non-affine function. 
\end{defn}

We will see in Section~\ref{sec: properties of tau} that a mass exponent $\tau$ of a sequence of discrete probability measures has the following three properties: 
\begin{enumerate}[label=\alph*)]
	\item $0\geq\tau\geq1-\id$, 
	\item $\tau$ is convex, 
	\item $\frac{\tau(q)}q \geq \tau'(q)$ for all $q\geq1$ where $\tau'(q)$ denotes the right-sided derivative.  
\end{enumerate}
Actually, we will see in Section~\ref{sec: char mass exp} that these properties characterize mass exponents as states the following theorem. 

\begin{thm}[Characterization of mass exponents] \label{thm: characterization}
	There is a one-to-one correspondence between functions satisfying properties \emph{a)-c)} and mass exponents of sequences of discrete probability measures. 
\end{thm}

This inverse problem asking which scaling laws can occur at all has also been addressed in the continuous setting~\cite{B15}. In the case of positive finite compactly supported Borel measures on $\R^d$, Barral characterized the arising mass exponents. The present result is close in spirit to Barral's theorem but the discrete setting imposes an additional constraint, reflecting the possible presence of large atoms. 

\subsection{Main Results}

We now state our main results. Recall that the parameter $\alpha$ was introduced in~(\ref{eq: ph delta type}). 

\begin{thm} \label{thm: graphs dynamical argument}
	Let $\alpha=0$ and $(X^n,L^n)_{n\in\N}$ be a sequence of star graphs such that the edge length vectors $L^n\in\R_+^n$ have rationally independent entries. Then, for any function $\tau$ satisfying the properties \emph{a)-c)} of mass exponents, there exists a sequence of mass measures whose mass exponent is $\tau$. 
\end{thm}

The theorem states that one can obtain any possible mass exponent via mass measures of quantum star graphs. 

The proof of Theorem~\ref{thm: graphs dynamical argument} uses a dynamical argument to extract a sequence of eigenvalues $k_n^2$ whose associated mass measures have mass exponent $\tau$. In order to be able to give the sequence of eigenvalues more explicitly, we will restrict to a class of star graphs that we call \emph{quasi-equilateral} whose edge lengths 
\begin{equation}
	L_i = 1+\ep\ell_i 
\end{equation}
depend on a parameter $\ep$ and vary on a lower scale. We assume these lengths to be increasing in~$i$. With~(\ref{eq: ph cont}) and~(\ref{eq: ph is sin}), the $\delta$-type condition~(\ref{eq: ph delta type}) is equivalent to 
\begin{equation} \label{eq: spectral equation cot alpha}
	\sum_{i=1}^n \cot kL_i = \frac\alpha k. 
\end{equation}
As a function of $k$, the sum of cotangents has poles at 
\begin{equation} \label{eq: poles of spectral determinant}
	p_{m,i} := \frac{m\pi}{L_i} \in \frac\pi{L_i} \Z. 
\end{equation}
By the mean value theorem, between two consecutive poles, there is exactly one zero. In the case of quasi-equilateral star graphs, the localization of eigenvalues can be made more precise as was found by Keating--Ueberschär. We restate their observations for the reader's convenience and refer to~\cite{KU} for more details. For a pole, we have 
\begin{equation}
	p_{m,i} = m\pi + O(m\ep\ell_i) = m\pi + O(m\ep\|\ell\|). 
\end{equation}
Then, the poles form clusters $\{p_{m,1},\dots,p_{m,N}\}$ labeled by $m\in\Z$ with 
\begin{equation}
	|p_{m,i} - p_{m,i'}| \asymp m\ep |\ell_i - \ell_{i'}|. 
\end{equation}
In the regime where 
\begin{equation} \label{eq: clustering regime}
	m\ep\|\ell\|\to0, 
\end{equation}
which we refer to as the \emph{clustering regime}, these clusters separate. More precisely, we have that $|p_{m,i} - p_{m',i'}| \gtrsim 1$ for $m\neq m'$. We note that $m$, $\ep$ and $\ell$ in general depend on $n$. The eigenvalues of the corresponding graph form clusters interlacing with the poles. This interlacing property is illustrated in Figure~\ref{fig: spectral determinant}. For an eigenvalue $k^2$ of $\Delta_L$ within a cluster, there exist unique $m\in\Z$ and $i_0\in\{1,\dots,n-1\}$ such that 
\begin{equation}
	p_{m,i_0+1} < k < p_{m,i_0}.
\end{equation}
We say that a sequence of eigenvalues $k_m^2$ satisfies the \emph{spectral repulsion condition} at position $i_0$ if there exists some $\delta\in(0,\frac12]$ such that for all $m$ 
\begin{equation} \label{eq: spectral repulsion}
	\frac{\min_{i=i_0,i_0+1} |k_m - p_{m,i}|}{|p_{m,i_0} - p_{m,i_0+1}|} 
		\geq \delta. 
\end{equation}

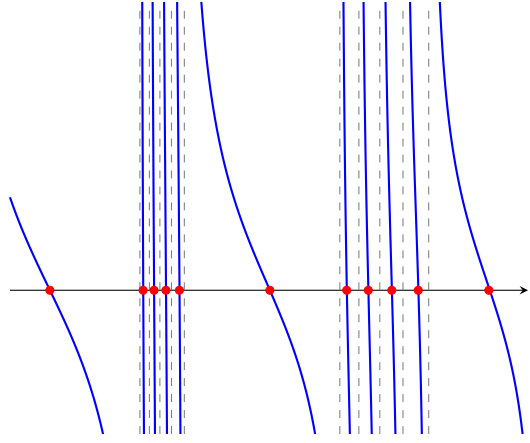
\begin{figure}
\begin{tikzpicture}
\begin{axis}[
    domain=0.5:4.2,
    samples=2000,
    ymin=-0.5, ymax=1,
    xmin=0.5, xmax=4.2,
    xtick=\empty,
    xticklabels=\empty,
    axis x line=middle,
    axis y line=none,
    enlargelimits=false,
    unbounded coords=jump,
]

\addplot[
    blue,
    thick,
]
{ 
  (abs(cot(deg(1.8*x))) > 100 ? nan : 0.1*cot(deg(1.8*x))) +
  (abs(cot(deg(1.9*x))) > 100 ? nan : 0.1*cot(deg(1.9*x))) +
  (abs(cot(deg(2*x)))   > 100 ? nan : 0.1*cot(deg(2*x))) +
  (abs(cot(deg(2.1*x))) > 100 ? nan : 0.1*cot(deg(2.1*x))) +
  (abs(cot(deg(2.2*x))) > 100 ? nan : 0.1*cot(deg(2.2*x)))
};

\foreach \a in {1.8,1.9,2,2.1,2.2} {
  \foreach \k in {1,...,10} {
    \pgfmathsetmacro{\xpole}{\k*pi/\a}
    \ifdim \xpole pt > 0.5pt
      \ifdim \xpole pt < 4.2pt
        \addplot[dashed, gray] coordinates {(\xpole,-0.5) (\xpole,1)};
      \fi
    \fi
  }
}

\addplot[
    only marks,
    mark=*,
    red,
    mark size=1.5pt,
]
coordinates {
    (0.785,0)
    (1.451,0)
    (1.529,0)
    (1.614,0)
    (1.710,0)
    (2.356,0)
    (2.905,0)
    (3.059,0)
    (3.227,0)
    (3.416,0)
    (3.92,0)
};

\end{axis}
\end{tikzpicture}
\caption{Sum of cotangents in~(\ref{eq: spectral equation cot alpha}) with zeros $k$ shown in red}
\label{fig: spectral determinant}
\end{figure}

For real-valued functions $f$ and $g$ depending in general on $n$, we write $f\lesssim g$ if there exists a positive constant $C>0$ independent of $n$ such that $f \leq C g$. If both $f\lesssim g$ and $g\lesssim f$ hold, we write $f\asymp g$. 

\begin{thm} \label{thm: quasi-equilateral stars}
	Let $\nu_n$ be a sequence of probability measures on $\{1,\dots,n\}$ wich satisfy $\nu_n(i)>\nu_n(i+1)$ for all $i$ and $i_0=i_0(n)\in\{1,\dots,n-1\}$. 
	We define 
	\begin{equation}
		L_i := 1+\ep\ell_i 
				\qquad\text{with}\qquad 
			\ell_i := \begin{cases}
					- \frac1{\sqrt{\nu_n(i)}}, 		& i<i_0, \\ 
					0, 								& i=i_0, \\ 
					\frac1{\sqrt{\nu_n(1)}}, 		& i=i_0+1 \\ 
					\frac1{\sqrt{\nu_n(1)}} 
						+ \frac1{\sqrt{\nu_n(i)}}, 	& i>i_0+1. 
				\end{cases} 
	\end{equation}
	Let $\alpha\in\R$ be such that the spectral repulsion condition~(\ref{eq: spectral repulsion}) is satisfied at $i_0$ for a sequence of eigenvalues $k^2 = k^2(n)$ such that $k(n)\to\infty$ as $n\to\infty$. Then, in the clustering regime~(\ref{eq: clustering regime}) we have that $\mu_n^k(i)\asymp\nu_n(i)$ for $i\neq i_0,i_0+1$. In particular, as $n\to\infty$, the measures~$\nu_n$ and the mass measures~$\mu_n^k$ have the same mass exponent.  
\end{thm}

The result should be compared to~\cite{KU}. It can be viewed as an extension of Corollary~2.2 therein, which, for $\ell_i = i^{(1-\gamma)/2}$ with $\gamma \in (0,1)$, determines the mass exponent 
\begin{equation}
	\tau(q) = \max\{1-q,-\gamma q\}, 
\end{equation}
which is multifractal in the sense of Definition~\ref{defn: multifractality}, along a sequence of eigenvalues in a fixed finite position in each cluster, i.e., $i_0$ does not depend on $n$. Our contribution generalizes this framework, while drawing substantially on the methods introduced in~\cite{KU}. We emphasize that in Theorem~\ref{thm: quasi-equilateral stars} one can obtain any possible mass exponent on quasi-equilateral star graphs. 

We comment on the assumption of choosing the coupling parameter $\alpha$ such that the spectral repulsion condition~(\ref{eq: spectral repulsion}) is satisfied. We first write $\alpha = \alpha_m(k_m)$ as a function of both the solution $k_m$ and the cluster label $m$. We show that in the clustering regime~(\ref{eq: clustering regime}), the coupling parameter $\alpha$, up to first order, only depends on the position of $k_m$ relative to its neighboring poles. More precisely, we show that up to first order $\alpha$ depends on $t\in(0,1)$ such that 
\begin{equation}
	k_m = k_m(t) = (1-t) p_{m,i_0} + t p_{m,i_0+1} 
		= m\pi \left(\frac{1-t}{L_{i_0}} + \frac t{L_{i_0+1}}\right). 
\end{equation}
For $m\ep\|\ell\|$ small, we Taylor expand the cotangent functions in (\ref{eq: spectral equation cot alpha}) and obtain 
\begin{equation}
	\alpha_m(k_m(t)) \asymp \sum_{i=1}^n \frac{k_m(t)}{k_m(t) - p_{m,i}}. 
\end{equation}
The factors $m\pi$ simplify showing that the right hand side is a rational function in $t$ independent of $m$. In particular, given $\delta\in(0,\frac12]$, there exists an interval depending only on $\delta$ from which we can choose the coupling parameter $\alpha$ in order to satisfy the spectral repulsion condition~(\ref{eq: spectral repulsion}). 

The paper is organized as follows. In the next section, we show general properties of the associated scaling exponents. Section~\ref{sec: char mass exp} proves the converse construction for abstract probability measures and characterizes the admissible class. Section~\ref{sec: spectral determinant} recalls the spectral determinant and the torus-flow description of the spectrum. Section~\ref{sec: semiclassical measures} relates semiclassical measures to edge-mass distributions and proves the realization theorem for generic star graphs. Section~\ref{sec: graph from measure} treats quasi-equilateral star graphs and gives the explicit construction realizing prescribed scaling laws.

\section{Properties of Mass Exponents} \label{sec: properties of tau}

In this section, we show the properties of the mass exponent already mentioned in Section~\ref{sec: set-up}. 

\begin{prop}[Properties of $\tau$] \label{prop: properties of tau}
	Let $\tau : [1,\infty)\longto\R$ be the mass exponent of a sequence of probability measures $\mu_n$. Then, 
	\begin{enumerate}[label=\alph*\emph{)}]
		\item $0\geq\tau\geq1-\id$, 
		\item $\tau$ is convex, 
		\item $\frac{\tau(q)}q \geq \tau'(q)$ for all $q\geq1$ where $\tau'(q)$ denotes the right-sided derivative. 
	\end{enumerate}
	Furthermore, $\tau$ is decreasing, $\tau'\geq-1$, and if $\lim\limits_{q\to\infty} \frac{\tau(q)}q = 0$, then $\tau=0$. 
\end{prop}

\begin{rem}
	The first two statements are standard and hold in various settings~\cite{PW97,O95,F03,BF13}. They characterize mass exponents in the case of positive, finite compactly supported Borel measures~\cite{B15}. The statement \emph{c}) seems to be specific to the discrete setting. One could replace the right-sided derivative with a left-sided derivative, except at the point $q=1$ where only one of the derivatives exists. Thus, for convenience, we state the proposition only for right-sided derivatives. 
\end{rem}

\begin{proof}
	We first show statement \emph{a}). By equivalence of norms, we have 
	\begin{equation}
		n^{1-q} \leq \|\mu_n\|_q^q \leq 1. 
	\end{equation}
	Taking the logarithm and dividing by $\log n$ gives the bounds on $\tau$. 
	
	Convexity follows by Hölder's inequality. For $q_1,q_2\geq1$ and $\lambda\in[0,1]$, we have 
	\begin{equation}
		\sum \mu_n(i)^{\lambda q_1 + (1-\lambda)q_2} 
			\leq \left(\sum \mu_n(i)^{q_1}\right)^\lambda 
					\left(\sum \mu_n(i)^{q_2}\right)^{1-\lambda}. 
	\end{equation}
	Taking logarithms, dividing by $\log n$ and passing to the limit superior, we obtain convexity of $\tau$, hence proving \emph{b}). 
	
	For \emph{c}), let $M_n = \max \mu_n$ be the maximal mass of $\mu_n$. We then have 
	\begin{equation}
		M_n^q \leq \|\mu_n\|_q^q \leq n M_n^q. 
	\end{equation}
	Taking logarithms and dividing by $\log n$, we obtain 
	\begin{equation}
		q \frac{\log M_n}{\log n} 
			\leq \frac{\log\|\mu_n\|_q^q}{\log n} 
			\leq q \frac{\log M_n}{\log n} + 1. 
	\end{equation}
	Denoting $C = \limsup_n \frac{\log M_n}{\log n} \leq 0$ and passing to the limit superior, we infer 
	\begin{equation} \label{eq: C bound on tau}
		q C \leq \tau(q) \leq q C + 1. 
	\end{equation}
	Since $\tau$ is convex, we also have 
	\begin{equation}
		\tau'(q)(p-q) + \tau(q) \leq \tau(p) 
	\end{equation}
	for all $p\geq1$. Dividing by $p$ and letting $p$ tend to infinity, we obtain 
	\begin{equation}
		\tau'(q) \leq \lim_{p\to\infty} \frac{\tau(p)}p 
			\leq \lim_{p\to\infty} \frac{pC+1}p = C. 
	\end{equation}
	Together with inequality~(\ref{eq: C bound on tau}), we therefore have 
	\begin{equation}
		\frac{\tau(q)}q \geq C \geq \tau'(q). 
	\end{equation}
	
	For the ``furthermore'' statement, $\tau$ is decreasing because otherwise one of its tangents would have positive slope forcing $\tau$ to be unbounded from above. Since $\tau$ is convex, its right-derivative $\tau'$ is increasing and we therefore have 
	\begin{equation}
		\tau' \geq \tau'(1) = \lim_{q\to1^+} \frac{\tau(q)-\tau(1)}{q-1} 
			\geq -1 
	\end{equation}
	where the last inequality is due to \emph{a}). Finally, assuming that $\lim_{q\to\infty}\frac{\tau(q)}q = 0$, inequality~(\ref{eq: C bound on tau}) yields both that $C = 0$ and $\tau\geq0$. Since $\tau\leq0$ by \emph{a}), we find $\tau=0$. 
\end{proof}

\begin{defn} \label{defn: P}
	We define the set 
	\begin{equation}
		\mathcal P 
			:= \left\{\tau : [1,\infty) \to \R \,\left|\, 
					0\geq\tau\geq1-\id, \tau \text{ is convex and } 
					\frac{\tau(q)}q \geq \tau'(q) \right.\right\} 
	\end{equation}
	of functions that satisfy the properties of mass exponents, where again $\tau'$ denotes the right-sided derivative of $\tau$. 
\end{defn}

We obtain the following corollary of Proposition~\ref{prop: properties of tau}. 

\begin{cor}[$\mathcal E \subset \mathcal P$] \label{cor: E subset P}
	Mass exponents are in $\mathcal P$. 
\end{cor}

Next, we show a stability result for the mass exponent that we will use in the following sections. 

\begin{lem} \label{lem: stability mass exponent inequality}
	Let $\mu_n, \nu_n$ be two sequences of probability measures on $\{1,\dots,n\}$, $I\subset\N$ a finite set and $j_0\in\N$. Assume that for all $q\geq1$ there exists a constant $C\geq1$ such that for all $n$ 
	\begin{equation}
		\|\mu_n\|_{\ell^q(I^c)} \leq C \|\nu_n\|_{\ell^q(I^c)} 
			\qquad\text{ and }\qquad 
		\mu_n(i) \leq C \nu_n(j_0), \,\text{ if } i\in I, 
	\end{equation}
	where $I^c$ is the complement of $I$ in $\{1,\dots,n\}$. Then, the mass exponent of $\mu_n$ is smaller than the mass exponent of $\nu_n$. 
\end{lem}

\begin{proof}
	We have that 
	\begin{equation}
		\|\mu_n\|_q^q 
			\leq C^q \sum_{i\notin I} \nu_n(i)^q 
				+ C^q \sum_{i\in I} \nu_n(j_0)^q 
			\leq C^q(1+\# I) \|\nu_n\|_q^q. 
	\end{equation}
	Taking logarithms and passing to the limit superior, we obtain the desired result. 
\end{proof}

\begin{cor}[Stability of the mass exponent] \label{cor: stability mass exponent}
	Let $\mu_n, \nu_n$ be two sequences of probability measures on $\{1,\dots,n\}$. If for all $q\geq1$ there exists $C\geq1$ such that for all $n$ 
	\begin{equation}
		C^{-1} \leq \frac{\|\mu_n\|_q}{\|\nu_n\|_q} \leq C, 
	\end{equation}
	then $\mu_n$ and $\nu_n$ have the same mass exponent. This condition is in particular satisfied if $\|\mu_n - \nu_n\|_1 \lesssim \frac1n$. 
\end{cor}

\begin{proof}
	Applying the previous lemma with $I=\varnothing$, we obtain the equality of mass exponents due to symmetry in $\mu_n$ and $\nu_n$. In order to show the sufficient criterion, we assume that $\|\mu_n - \nu_n\|_1 \lesssim \frac1n$. We then have 
	\begin{equation}
		\left|1-\frac{\|\nu_n\|_q}{\|\mu_n\|_q}\right| 
			\leq \frac{\|\mu_n - \nu_n\|_q}{\|\mu_n\|_q} 
			\leq \frac{\|\mu_n - \nu_n\|_1}{n^{\frac1q -1}} 
			\lesssim n^{-\frac1q} < 1 
	\end{equation}
	for large enough $n$, where, in the second step, we bounded the $q$-norm by the $1$-norm and used the fact that $\|\mu_n\|_q \geq n^{\frac1q -1}$ since $\|\mu_n\|_1=1$. This finishes the proof. 
\end{proof}

\section{Characterization of Mass Exponents} \label{sec: char mass exp}

After having proven the three properties \emph{a})-\emph{c}) of mass exponents, we address here the question of finding a sequence of probability measures whose mass exponent is some prescribed function satisfying the properties \emph{a})-\emph{c}). The idea is to start with piecewise affine functions that can be obtained exactly as mass exponents of blockwise uniform measures and then pass to convex functions via a diagonal sequence argument. 

\begin{defn} \label{defn: T_fin}
	We define the set,
	\begin{equation}
		\mathcal T_\mathrm{fin} 
			:= \left\{\left.\tau=\max_{1 \leq j \leq m} f_j 
						\,\right|\, \tau(1)=0, f_j \in \mathcal A\right\}
	\end{equation}
	of pointwise maxima of finitely many affine functions in 
	\begin{equation}
		\mathcal A 
			:= \{f:[1,\infty)\to\R \mid f(q) = -aq+c, 0 \leq c \leq a \leq 1\}.
	\end{equation}
\end{defn}

The next definition will be of importance in the study of mass measures. 

\begin{defn} \label{defn: M}
	We define 
	\begin{equation}
		\mathcal M 
			:= \left\{\mu\in\R_+^n \,\left|\, \sum \mu_i = 1 
				\text{ and } \sum \ep_i \sqrt{\mu_i} = 0 
				\text{ for some } \ep\in\{\pm1\}^n\right.\right\}. 
	\end{equation}
\end{defn}

The succeeding lemma states that elements of $\mathcal T_\mathrm{fin}$ can be exactly attained as mass exponents.  

\begin{lem}[$\mathcal T_\mathrm{fin} \subset \mathcal E$] \label{lem: T_fin subset E}
	For $\tau\in\mathcal T_\mathrm{fin}$, there exists a sequence $\mu_n$ of probability measures such that 
	\begin{itemize}[label = --]
		\item the mass exponent of $\mu_n$ is $\tau$, 
		\item the sequence of functions 
		\begin{equation}
			\tau_n(q) = \frac{\log\|\mu_n\|_q^q}{\log n}, \qquad q\geq1
		\end{equation}
			converges uniformly to $\tau$ on compact sets, 
		\item there exists $\ep\in\{\pm1\}^n$ with $\sum_i \ep_i\sqrt{\mu_n(i)} = 0$, i.e., $\mu_n\in\mathcal M$.  
	\end{itemize}
\end{lem}

\begin{proof}
	Let $\tau = \max_{1 \leq j \leq m} f_j\in\mathcal T_\mathrm{fin}$. As a piecewise affine function, there exist $m$ intervals $I_j$ covering $[1,\infty)$ and such that $\tau$ coincides with $f_j$ on $I_j$. We may assume that $0\in I_1$ and $\sup I_j = \inf I_{j+1}$. We write $f_j(q) = -a_j q + c_j$. Since $f_1(1)=\tau(1)=0$, we have $c_1=a_1$ and $c_j<a_j$ for $j>1$. In particular, only $c_1$ can take the value $1$. We partition $\{1,\dots,n\}$ into $m+1$ boxes that we will denote $B_j$ that will satisfy the following properties: 
	\begin{itemize}[label = --]
		\item $\# B_j\asymp n^{c_j}$ for $1\leq j\leq m$, 
		\item $\# B_j$ is even for $1\leq j\leq m$, 
		\item $\# B_{m+1}\geq2$. 
	\end{itemize}
	More precisely, we make the following construction. If $c_1\neq1$, we choose 
	\begin{equation}
		\# B_j 
			= 2\left\lceil\frac{n^{c_j}}2\right\rceil, \quad 1\leq j\leq m, 
	\end{equation}
	that is, $\# B_j$ is the smallest even number larger than $n^{c_j}$. If $c_1=1$, we choose $\# B_j$ with $2\leq j\leq m$ as before, and put 
	\begin{equation}
		\# B_1 
			= 2\left\lfloor\frac{n-2}2\right\rfloor - \sum_{j=2}^m \# B_j. 
	\end{equation}
	Independent of $c_1$, we define $B_{m+1}$ to be the complement of $\bigcup B_j$. Notice that if $n$ is large enough, this construction yields the three properties we mentioned before. 
	
	Now, we define the measure $\mu_n$ blockwise via 
	\begin{equation}
		\mu_n(i) = \begin{cases}
					\frac1{C_n} n^{-a_j},	& i \in B_j, 1 \leq j \leq m, \\
					\frac1{C_n} n^{-2}, 	& i \in B_{m+1}, 
					\end{cases} 
	\end{equation}
	where the normalization constant $C_n$ guarantees that the $\mu_n$ are probability measures. One easily verifies that $C_n \to 1^+$ as $n\to\infty$. For the moment sum, we obtain 
	\begin{equation}
		C_n^q \sum_{i=1}^m \mu_n(i)^q 
			= \sum_{j=1}^m \# B_j n^{-qa_j} + \# B_{n+1} n^{-2q} 
			= \sum_{j=1}^m n^{f_j(q)} \frac{\# B_j}{n^{c_j}} 
				+ O(n^{1-2q}). 
	\end{equation}
	We note that $c^{-1}\leq\# B_j/n^{c_j}\leq c$ for some constant $c>1$. The dominating term has the power $\max f_j(q)=\tau(q)$. This shows that the mass exponent of $\mu_n$ is $\tau$. 
	
	For the uniform convergence on compact sets, we fix some $q_0\geq1$. For large enough $n$, we have $1 \leq C_n \leq 2$ since $C_n \to 1^+$. By the computation above, we thus on one hand obtain 
	\begin{equation}
		\tau_n(q) \leq \frac{\log(C_n^{-q}(mc+1)n^{\max f_j(q)})}{\log n} 
			\leq \tau(q) + \frac{\log(mc+1)}{\log n} 
	\end{equation}
	for all $q\geq1$. On the other hand, we have 
	\begin{equation}
		\tau_n(q) \geq \frac{\log(2^{-q} c^{-1} n^{\max f_j(q)})}{\log n} 
			\geq \tau(q) - q\frac{\log2c}{\log n}. 
	\end{equation}
	These two bounds together yield 
	\begin{equation}
		\|\tau_n-\tau\|_{\infty,[1,q_0]} 
			\leq q_0 \frac{\log2mc}{\log n} \to 0. 
	\end{equation}
	
	We are left to show that the measures $\mu_n$ can be chosen in $\mathcal M$. We notice that changing the measure on $B_{m+1}$ on only a fixed finite subset, does not affect the mass exponent. Since $\# B_j$ is even, we infer that $\# B_{m+1}$ is even if and only if $n$ is even. Then, if $n$ is even, since $\mu_n$ is uniform on each box, we can choose alternating signs for cancellation, thus achieving $\mu_n\in\mathcal M$. Otherwise, we change the last two masses to 
	\begin{equation}
		\tilde\mu_n(n-1) = \tilde\mu_n(n) = \frac{\mu_n(n)}2 
	\end{equation}
	and renormalize. In this way, together, $\tilde\mu_n(n-1)$ and $\tilde\mu_n(n)$ can cancel out $\mu_n(n-2)$, thereby concluding the proof. 
\end{proof}

\begin{lem}[$\overline{\mathcal T}_\mathrm{fin}^\mathrm{\,pt} \subset \mathcal E$] \label{lem: T_fin^pt subset E}
	If $\tau\in\overline{\mathcal T}_\mathrm{fin}^\mathrm{\,pt}$ is a pointwise limit of elements in $\mathcal T_\mathrm{fin}$, then there exists a sequence of probability measures whose mass exponent is $\tau$. The probability measures can be chosen in $\mathcal M$. 
\end{lem}

\begin{proof}
	The proof is a standard diagonal-sequence argument. Let $\ep>0$ and $\tau_m\in\mathcal T_\mathrm{fin}$ such that $\tau_m$ converges pointwise to $\tau$. By Lemma~\ref{lem: T_fin subset E}, for every $m$, there is a sequence $(\mu_{m,n})_n$ of measures in $\mathcal M$ with mass exponent $\tau_m$. More specifically, due to the compact convergence in Lemma~\ref{lem: T_fin subset E}, we can extract a subsequence $n_m$ such that 
	\begin{equation}
		\left\|\frac{\log\sum \mu_{m,n_m}^{(\cdot)}}{\log n_m} 
				- \tau_m\right\|_{\infty,K_\ell} 
			< \frac\ep2 
	\end{equation}
	for all $\ell=1,\dots,m$ where $K_\ell$ is an exhausting sequence of compact sets of $[1,\infty)$. We define 
	\begin{equation}
		\nu_m = \mu_{m,n_m}.
	\end{equation}
	Then, for $q\geq1$, there exists $\ell_0$ such that $q\in K_{\ell_0}$ as well as $m_0$ such that 
	\begin{equation}
		|\tau_m(q) - \tau(q)| < \frac\ep2
	\end{equation}
	if $m>m_0$. Now for $m$ larger than $\ell_0$ and $m_0$, we conclude that 
	\begin{equation}
		\left|\frac{\log\sum \nu_m^q}{\log n_m} - \tau(q)\right| 
			\leq \left\|\frac{\log\sum \nu_m^{(\cdot)}}{\log n_m} 
				- \tau_m\right\|_{\infty,K_{\ell_0}} + |\tau_m(q) - \tau(q)| 
			< \ep. 
	\end{equation}
	This finishes the proof.
\end{proof}

\begin{defn} \label{defn: T}
	We denote by 
	\begin{equation}
		\mathcal T := \left\{\left.\tau=\sup_{j\in J} f_j \,\right|\, 
										\tau(1)=0, f_j \in \mathcal A, 
										J \textnormal{ countable}\right\}
	\end{equation}
	the set of pointwise suprema of sequences in $\mathcal A$. 
\end{defn}

\begin{lem}[$\mathcal P \subset \mathcal T$] \label{lem: P subset T}
	Functions satisfying the properties of mass exponents shown in Proposition~\ref{prop: properties of tau} are pointwise suprema of sequences in $\mathcal A$. 
\end{lem}

\begin{proof}
	Let $\tau\in\mathcal P$. Then, by Definition~\ref{defn: P}, $\tau$ is convex and satisfies the bounds $0\geq\tau\geq1-\id$ and $\frac{\tau(q)}q \geq \tau'(q)$, $q\geq1$. It is a fact from convex analysis that due to convexity of $\tau$, we can write 
	\begin{equation}
		\tau(q) = \sup\{T_p(q) \mid p\in[1,\infty)\cap\Q\} 
	\end{equation}
	where $T_p(q)=\tau'(p)(q-p)+\tau(p)$ denotes the tangent of $\tau$ at $p$ and $\tau'$ is again the right-derivative of $\tau$. We are left to show that the tangents $T_p$ lie in $\mathcal A$. To that end, we write $T_p(q) = -a_p q + c_p$ and show the three inequalities $0 \leq c_p \leq a_p \leq 1$. As $\tau$ is negative, we have $0\geq T_p(1) = -a_p + c_p$. Thus, $c_p \leq a_p$. Furthermore, we have that $\tau'$ is increasing due to convexity which implies that $-a_p = \tau'(p) \geq \tau'(1) \geq -1$, i.e., $a_p\leq1$. At last, since $\frac{\tau(q)}q \geq \tau'(q)$, we have 
	\begin{equation}
		c_p = \tau(p) - p \tau'(p) \geq 0. 
	\end{equation}
	This finishes the proof. 
\end{proof}

\begin{thm}[Characterization of mass exponents] \label{thm: characterization E=P=T}
	The sets $\mathcal E$, $\mathcal P$, $\mathcal T$, $\overline{\mathcal T}_\mathrm{fin}^\mathrm{\,pt}$ and $\overline{\mathcal T}_\mathrm{fin}^\mathrm{\,c.o.}$ all coincide, where the latter is the closure of $\mathcal T_\mathrm{fin}$ with respect to the compact-open topology. In particular, there is a one-to-one correspondence between functions satisfying properties a)-c) from the introduction and mass exponents of sequences of discrete probability measures. 
\end{thm}

\begin{proof}
	We have the following chain of inclusions: 
	\begin{equation}
		\overline{\mathcal T}_\mathrm{fin}^\mathrm{\,pt} 
			\subset \mathcal E \subset \mathcal P \subset \mathcal T 
			\subset \overline{\mathcal T}_\mathrm{fin}^\mathrm{\,c.o.}. 
	\end{equation}
	The first inclusion was shown in Lemma~\ref{lem: T_fin^pt subset E}, the second one is Corollary~\ref{cor: E subset P}, the third is due to Lemma~\ref{lem: P subset T} and the fourth one is shown as follows. Let $\tau = \sup_\N f_j$. We set $\tau_m = \max_{j \leq m} f_j \in \mathcal T_\mathrm{fin}$. Then, $\tau_m$ is an increasing sequence converging pointwise to $\tau$. By Dini's Theorem, $\tau_m$ converges uniformly to $\tau$ on compact sets. Now, the statement follows since the compact-open topology is finer than the topology of pointwise convergence. 
\end{proof}

\begin{cor} \label{cor: E attainable via M}
	Any mass exponent can be attained by measures in $\mathcal M$. 
\end{cor}

\begin{proof}
	Let $\tau\in\mathcal E$ be a mass exponent. Then, by Theorem~\ref{thm: characterization E=P=T}, $\tau\in\overline{\mathcal T}_\mathrm{fin}^\mathrm{\,pt}$. With Lemma~\ref{lem: T_fin^pt subset E}, we can conclude that $\tau$ is mass exponent of measures in $\mathcal M$. 
\end{proof}

We finish this section with an example of a strictly convex mass exponent. The example is an adaptation of \cite[Example 17.1]{F03}. 

\begin{eg}
	We fix $p\in[0,\frac12]$ and define a sequence of measures $\mu_{2^n}$ on $\mathcal P(\{1,\dots,n\})$ via 
	\begin{equation}
		\mu_{2^n}(S) = p^{\# S}(1-p)^{n-\# S} 
	\end{equation}
	for $S\subset\{1,\dots,n\}$. Then, each $\mu_{2^n}$ is indeed a probability measure by the Binomial Theorem. For its moments, we have 
	\begin{equation}
		\sum_{S\subset\{1,\dots,n\}} \mu_{2^n}(S)^q 
			= \sum_{j=1}^n {n \choose j} (p^q)^j((1-p)^q)^{n-j} 
			= (p^q + (1-p)^q)^n. 
	\end{equation}
	We thus find 
	\begin{equation}
		\tau(q) 
			= \lim_{n\to\infty} \frac{\log \sum_S \mu_{2^n}(S)^q}
									{\log 2^n} 
			= \lim_{n\to\infty} \frac{n \log(p^q + (1-p)^q)}{n\log 2} 
			= \frac{\log(p^q + (1-p)^q)}{\log 2}. 
	\end{equation}
	In particular, we notice that $\tau$ is affine only for $p\in\{0,\frac12\}$ and otherwise strictly convex. Hence, the constructed sequences of measures are multifractal for all $p\in(0,\frac12)$. 
\end{eg}

\section{Spectral Determinant} \label{sec: spectral determinant}

In this section, we introduce the spectral determinant that will be a central tool in the proof of Theorem~\ref{thm: graphs dynamical argument}. We have already observed in~(\ref{eq: ph is sin}) that eigenfunctions are of the form $\ph^{\langle k\rangle}_i(x) = \alpha_i \sin kx$. Due to continuity~(\ref{eq: ph cont}) and the $\delta$-type condition~(\ref{eq: ph delta type}), these coefficients $\alpha_i$ give rise to an eigenfunction $\ph^{\langle k\rangle}$ of $\Delta_L$ if and only if $(\alpha_i)_{1\leq i\leq n}$ is a non-trivial element of the kernel of $M(kL,\frac\alpha k)$ where 
\begin{equation}
	M(x,y) 
		= \begin{pmatrix}
			\sin x_1 & -\sin x_2 & 0 & \cdots & 0 \\ 
			0 & \sin x_2 & -\sin x_3 & 0 & \vdots \\ 
			\vdots & \ddots & \ddots & \ddots & \vdots \\ 
			0 & \cdots & 0 & \sin x_{n-1} & -\sin x_n \\ 
			\cos x_1 & \cos x_2 & \cdots & \cdots & \cos x_n - y \sin x_n
		\end{pmatrix} 
\end{equation}
for $x\in\R^n$ and $y\in\R$. A Laplace expansion with respect to the last row shows that 
\begin{equation} \label{eq: determinant M(x,y)}
	\det M(x,y) = \sum_{j=1}^n \cos x_j \prod_{i\neq j} \sin x_i 
						- y \prod_{i=1}^n \sin x_i 
\end{equation}
which we will refer to as the \emph{spectral determinant} of $X^n$. It allows to write the spectrum of~$\Delta_L$ as 
\begin{equation} \label{eq: spectrum as vanishing locus}
	\sigma(\Delta_L) = \{k^2\mid k\neq0, \det M(kL,\tfrac\alpha k)=0\}. 
\end{equation}
We show that if $L$ has rationally independent entries, then $\ph^{\langle k\rangle}(v)\neq0$. We argue via the contraposition. If $\ph^{\langle k\rangle}$ vanishes at the central vertex~$v$, then due to continuity, 
\begin{equation}
	0 = \ph^{\langle k\rangle}(v) = \ph^{\langle k\rangle}_i(L_i) 
		= \alpha_i \sin kL_i
\end{equation}
for all $i$. Hence, $kL_i\in\pi\Z$ for all $i$ such that $\alpha_i\neq0$. Noticing that none of the columns of $M(kL,\frac\alpha k)$ vanishes, there exist at least two edges $i,i'$ with $\alpha_i,\alpha_{i'}\neq0$. Thus, there are rational dependencies in the entries of $L$. Therefore, if $L$ has rationally independent entries, then $\ph^{\langle k\rangle}$ does not vanish at the central vertex. 

In \cite{BG00}, Barra--Gaspard observed that, if the coupling constant $\alpha$ vanishes, the spectrum of $\Delta_L$ can be expressed with the help of a dynamical system on the $n$-dimensional torus where $n$ is the number of edges of $X^n$. Fruits of this approach were harvested in many works including \cite{KMW03,CdV15,RR20}. We will repeat parts of their argument here. So for the rest of this section, we take $\alpha=0$. In this case, the determinant in~(\ref{eq: determinant M(x,y)}) descends to the quotient $\T^n = \R^n/2\pi\Z^n$ where we define the function 
\begin{equation} \label{eq: spectral determinant}
	F : \T^n \longto \R, \, x \longmapsto \det M(x,0). 
\end{equation}
Notice that the spectral determinant only depends on the topology of the star graph, not on its geometry, i.e., $F$ does not depend on $L$. We denote the subset of $\T^n$ given as the vanishing locus of $F$ by 
\begin{equation}
	Z = \{x\in\T^n \mid F(x)=0\}. 
\end{equation}
The spectrum of the Laplacian $\Delta_L$ is then given by 
\begin{equation}
	\sigma(\Delta_L) = \{k^2\mid k>0, [kL]\in Z\} 
\end{equation}
with the projection map $[\,\cdot\,] : \R^n\longto\T^n$. The linear flow $t \longmapsto [tL]$ on the torus induced by $L$ intersects $Z$ exactly at the times $k$ such that $k^2$ is an eigenvalue of $\Delta_L$. We furthermore note that, by the Kronecker--Weyl Equidistribution Theorem, if the entries of $L$ are rationally independent, the flow is ergodic. In particular, in that case the intersection points of the linear flow generated by $L$ are dense in $Z$. On the set $A = \{x\in\T^n\mid x_i\neq 0\mod\pi \text{ for all } i\}$ of \emph{admissible} points, we can then write 
\begin{equation} \label{eq: spectral determinant factorization}
	F(x) = \left(\prod_{j=1}^n \sin x_j\right) 
				\sum_{i=1}^n \frac{\cos x_i}{\sin x_i}. 
\end{equation}
Thus, in the coordinates $y_i = \cot x_i$, we obtain 
\begin{equation}
	x\in Z \cap A 
		\iff \sum_{i=1}^n \cot x_i = \sum_{i=1}^n y_i = 0 
		\iff y\in(1,\dots,1)^\perp. 
\end{equation}
With the transformation 
\begin{equation} \label{eq: coordinates arccot}
	T : \R^n \longto [(0,\pi)^n], \, 
		y_i\longmapsto x_i = \operatorname{arccot} y_i
\end{equation}
this means that $T((1,\dots,1)^\perp) = Z\cap[(0,\pi)^n] \subset Z \cap A$. 

We observe that if $k^2$ is an eigenvalue of $\Delta_L$ and $L$ is such that $[kL]\notin A$, that is $kL_i\in\pi\Z$ for some $i$, then the eigenfunction $\ph^{\langle k\rangle}$ vanishes at the central vertex $v$. We have seen above that this implies rational dependencies in the entries of $L$. Therefore, if $L$ has rationally independent entries, then for all eigenvalues $k^2\in\sigma(\Delta_L)$, we have $[kL]\in Z\cap A$.

\section{Semiclassical Measures and Mass Measures} \label{sec: semiclassical measures}

In this section, we first show that mass measures and semiclassical measures are related as in~(\ref{eq: relation mass and semiclassical measures}). We then show a lemma that gives an approximation of mass measures via measures in $\mathcal M$ from Definition~\ref{defn: M} before proving Theorem~\ref{thm: graphs dynamical argument}. 

For an eigenvalue $k^2$ of $\Delta_L$ where $L$ has rationally independent entries, one easily verifies that the associated eigenfunction $\ph^{\langle k\rangle}$ is of the form 
\begin{equation} \label{eq: formula for eigenfunctions}
	\ph^{\langle k\rangle}_i(x) 
		= \ph^{\langle k\rangle}(v) \frac{\sin kx}{\sin kL_i} 
\end{equation}
on the $i$-th edge where $v$ denotes the central vertex. As normalization, we choose 
\begin{equation} \label{eq: normalization of eigenfuncitons}
	|\ph^{\langle k\rangle}(v)|^2 \sum_{i=1}^n \frac1{\sin^2 kL_i} = 1. 
\end{equation}
With this normalization, the associated mass measure $\mu_n^k$ will be a probability measure. 

\begin{lem}[Classification of semiclassical measures] \label{lem:class sc measures}
	Let $\alpha=0$ and fix a star graph $X^n$. For $L\in\R_+^n$ with rationally independent entries, the set normalized semiclassical measures is 
	\begin{equation}
		\left\{\left.\frac12 \sum_{j=1}^n \mu_i \mathrm{Leb}_i \,\right|\, 
			(\mu_1,\dots,\mu_n)\in\overline{\operatorname{im} G}\right\} 
	\end{equation}
	where the measures are subject to the normalization~(\ref{eq: normalization of eigenfuncitons}) of eigenfunctions and $G$ is the map $G : Z\cap A \longto \S^{n-1}_{\|\cdot\|_1}$ via 
	\begin{equation}
		x \longmapsto \frac1{\sum_i \sin^{-2} x_i}
				\left(\frac1{\sin^2 x_1},\dots,\frac1{\sin^2 x_n}\right). 
	\end{equation}
	More precisely, the weak*-limit $\mu_\mathrm{sc}$ of $|\ph^{\langle k_m\rangle}|^2 \mathrm{vol}_L$ exists if and only if $\mu_n^{k_m}$ converges, and we have 
	\begin{equation}
		\mu_\mathrm{sc} 
			= \frac12 \sum_{i=1}^n \left(\lim_{m\to\infty} \mu_n^{k_m}(i)\right) \mathrm{Leb}_i. 
	\end{equation}
\end{lem}

\begin{rem}
	The previous lemma is a special case of the classification result which is due to Colin~de~Verdière~\cite{CdV15} for general quantum graphs. The description features the Gauß map which associates to each regular point of $Z$ the unit normal vector $\nu_x$ at $x$. As it turns out, all admissible points $x\in Z\cap A$ are the regular points of $Z$ and the unit normal vector 
	\begin{equation}
		\nu_x = -\frac{\nabla F(x)}{\|\nabla F(x)\|_1} 
	\end{equation}
	is precisely given by $G(x)$. 
\end{rem}

\begin{proof}
	An eigenfunction $\ph^{\langle k\rangle}$ induces the measure $|\ph^{\langle k\rangle}|^2 \operatorname{vol}_L$. For a continuous function $f$ on $(X,L)$, we compute with~(\ref{eq: formula for eigenfunctions}) 
	\begin{equation}
		|\ph^{\langle k\rangle}|^2 \operatorname{vol}_L (f) 
			= \sum_{i=1}^n \int_0^{L_i} f(x) |\ph^{\langle k\rangle}_i(x)|^2 dx 
			= \sum_{i=1}^n \frac{|\ph^{\langle k\rangle}(v)|^2}{\sin^2 kL_i} 
					\int_0^{L_i} f(x) \sin^2 kx \,dx.  
	\end{equation}
	The integral converges to $\frac12 \int_0^{L_i} f(x) dx = \frac12 \mathrm{Leb}_i(f)$ as $k\to\infty$ due to Riemann--Lebesgue. Considering functions $f$ whose integral on one edge is 1 and that vanish on the other edges, we find that along a subsequence $k_m^2$ of eigenvalues the measures $|\ph^{\langle k_m\rangle}|^2 \operatorname{vol}_L$ have a weak*-limit if and only if $G(k_m L)$ converges. The limit is then given by 
	\begin{equation}
		\lim_{m\to\infty} |\ph^{\langle k_m\rangle}|^2 \operatorname{vol}_L 
			= \sum_{i=1}^n \left(\lim_{m\to\infty} 
				\frac{|\ph^{\langle k_m\rangle}(v)|^2}{\sin^2 k_m L_i}\right) 
				\frac12 \,\mathrm{Leb}_i 
			= \frac12 \sum_{i=1}^n \left(\lim_{m\to\infty} G_i(k_m L)\right) 
				\mathrm{Leb}_i 
	\end{equation}
	where we used the normalization~(\ref{eq: normalization of eigenfuncitons}) and the definition of $G$. We are left to show that for a point $[kL]\in Z\cap A$, the coordinates of $G([kL])$ give rise to the mass measure $\mu_n^k$. Indeed, with the formula for the coordinate functions of $\ph^{\langle k\rangle}$ from~(\ref{eq: formula for eigenfunctions}) and the normalization of eigenfunctions introduced in~(\ref{eq: normalization of eigenfuncitons}), we have 
	\begin{equation} \label{eq: mass measure via G}
		\|\ph_i^{\langle k\rangle}\|_\infty^2 
			=  |\ph_i^{\langle k\rangle}(v)|^2 \frac1{\sin^2 kL_i} 
			= \frac{\sin^{-2} kL_i}{\sum_j \sin^{-2} kL_j} 
			= G_i([kL]). 
	\end{equation}
	With the definition of mass measure in~(\ref{eq: mass measure}), we find $\mu_n^k(i) = G_i([kL])$. This finishes the proof. 
\end{proof}

In the coordinates $T$ we introduced in (\ref{eq: coordinates arccot}), $G$ takes the form 
\begin{equation}
	G(T(y)) = \frac{(1+y_1^2,\dots,1+y_n^2)}{n+\|y\|_2^2} 
\end{equation}
due to the trigonometric identity $\sin(\operatorname{arccot} t) = (1+t^2)^{-1/2}$. 

Before proving Theorem~\ref{thm: graphs dynamical argument}, we show the following lemma. It essentially states that measures in the set $\mathcal M$ from Definition~\ref{defn: M} lie close to mass measures. 

\begin{lem}[Approximation of mass measures] \label{lem: approximate G via M}
	For every $\mu\in\mathcal M$ and $\eta>0$, there exist an admissible point $x\in Z\cap A$ such that 
	\begin{equation}
		\|G(x)-\mu\|_1 \leq \eta. 
	\end{equation}
\end{lem}

\begin{proof}
	Set $\gamma := \sqrt{\eta/(2n)}$. We start by noting that due to the definition of $\mathcal M$, we have that $\mu_i>0$ and there exists $\ep\in\{\pm1\}^n$ such that $\sum_i \ep_i \sqrt{\mu_i} = 0$. Hence, defining 
	\begin{equation}
		y_i := \ep_i \sqrt{\mu_i}, 
	\end{equation}
	we have both $\sum_i y_i = 0$ and $\|y\|_2=\|\mu\|_1=1$. Then setting $x=T(\frac y\gamma)$, we have 
	\begin{equation}
		x\in T((1,\dots,1)^\perp) = Z\cap[(0,\pi)^n] \subset Z\cap A 
	\end{equation}
	as we have noticed before. We now compute the $i$-th component of $G(x)$ and its distance to $\mu_i$. We have 
	\begin{equation}
		G_i(x) = G_i(T(\tfrac y\gamma)) 
			= \frac{1+\frac{\mu_i}{\gamma^2}}{n+\frac1{\gamma^2}} 
			= \frac{\gamma^2+\mu_i}{n\gamma^2+1}. 
	\end{equation}
	For the distance to $\mu_i$, we find 
	\begin{equation}
		|G_i(x) - \mu_i| 
			\leq \left|\gamma^2+\mu_i - (n\gamma^2+1)\mu_i\right| 
			\leq (1+n\mu_i)\gamma^2 
			= \left(\frac1n + \mu_i\right)\frac\eta2. 
	\end{equation}
	Thus, summing over $i$ yields $\|G(x)-\mu\|_1 \leq \eta$ which finishes the proof. 
\end{proof}

We are now ready to prove Theorem~\ref{thm: graphs dynamical argument} which establishes that any possible mass exponent can be produced as a sequence of mass measures on star graphs with arbitrary rationally independent edge lengths. 

\begin{proof}[Proof of Theorem~\ref{thm: graphs dynamical argument}]
	We have $(X^n,L^n)$ a sequence of star graphs with rationally independent edge lengths prescribed by $L^n$ as well as a function $\tau\in\mathcal P$, i.e., $\tau$ satisfying the properties \emph{a})-\emph{c}) of a mass exponent. We show that there exists a sequence $k_n^2$ of eigenvalues of $\Delta_{L^n}$ whose associated sequence of mass measures on $(X^n,L^n)$ has mass exponent $\tau$. 
	
	Due to the characterization of mass exponents and Corollary~\ref{cor: E attainable via M}, we can find a sequence of measures $\mu_n\in\mathcal M$ whose mass exponent is given by $\tau$. Furthermore, by Lemma~\ref{lem: approximate G via M}, there exist $x_n\in Z\cap A$ such that 
	\begin{equation} \label{eq: in proof of main 1}
		\|G(x_n) - \mu_n\|_1 \leq \frac1{2n}. 
	\end{equation}
	Since the entries of $L^n$ are rationally independent, the flow $t\longmapsto[tL^n]$ is ergodic. Furthermore, $L^n$ is transverse to $Z\cap A$~\cite[Lemma 3.3]{RR20}. Thus, for any $\ep>0$, there exist (in fact infinitely many) $k_n>0$ such that $[k_n L^n] \in B_\ep(x_n)\cap Z\cap A$. In particular, $k_n^2$ is an eigenvalue of $\Delta_{L^n}$. Since $G$ is continuous on $Z\cap A$, we find, upon shrinking $\ep$ if necessary, that 
	\begin{equation} \label{eq: in proof of main 2}
		\|G(x_n) - G([k_n L^n])\|_1 \leq \frac1{2n}. 
	\end{equation}
	Since, under the identification~(\ref{eq: identification set of measures and conv}), $G([k_n L^n]) = \mu_n^{k_n}$, equations (\ref{eq: in proof of main 1}) and (\ref{eq: in proof of main 2}) yield 
	\begin{equation}
		\|\mu_n^{k_n} - \mu_n\|_1 \leq \frac1n. 
	\end{equation}
	Now from Corollary~\ref{cor: stability mass exponent}, we get that the mass exponents of $\mu_n^{k_n}$ and $\mu_n$ coincide, that is that the mass exponent of $\mu_n^{k_n}$ is $\tau$, thereby finishing the proof. 
\end{proof}

\section{From Measures to Graphs} \label{sec: graph from measure}

The goal of this section is to prove Theorem~\ref{thm: quasi-equilateral stars}. We recall that for a star graph $(X,L)$ with rationally independent edge lengths, by (\ref{eq: spectrum as vanishing locus}), we have that $k^2>0$ is an eigenvalue of $\Delta_L$ if and only if $\sum_i \cot kL_i - \frac\alpha k$ vanishes. As a function of $k$, this sum has poles at 
\begin{equation}
	p_{m,i} := \frac{m\pi}{L_i} \in \frac\pi{L_i} \Z. 
\end{equation}
The distance between two poles can be estimated as 
\begin{equation} \label{eq: poles estimate}
	|p_{m,i} - p_{m,i'}| \asymp m\ep|\ell_i - \ell_{i'}|.
\end{equation}
We now turn to the proof of Theorem~\ref{thm: quasi-equilateral stars}. 

\begin{proof}[Proof of Theorem~\ref{thm: quasi-equilateral stars}]
	By Taylor's theorem, we have 
	\begin{equation}
		\sin^{-2} kL_i = \sin^{-2}(k-p_{m,i})L_i \asymp |k-p_{m,i}|^{-2}. 
	\end{equation}
	If $i=i_0,i_0+1$, due to the spectral repulsion condition~(\ref{eq: spectral repulsion}) and the estimate for the distance between poles in~(\ref{eq: poles estimate}), we have 
	\begin{equation}
		|k-p_{m,i}| 
			\asymp |p_{m,i_0}-p_{m,i_0+1}| 
			\asymp m\ep|\ell_{i_0}-\ell_{i_0+1}| 
			= \frac{m\ep}{\sqrt{\nu_n(1)}}. 
	\end{equation}
	If $i\neq i_0,i_0+1$, we can estimate the distance from $k$ to a pole $p_{m,i}$ from below and above via the distance from one of its neighboring poles $p_{m,i_0}$ and $p_{m,i_0+1}$. More specifically, due to~(\ref{eq: poles estimate}) for $i<i_0$, we have 
	\begin{equation}
		m\ep|\ell_{i_0} - \ell_i| 
			\asymp |p_{m,i_0} - p_{m,i}| 
			\leq |k - p_{m,i}| 
			\leq |p_{m,i_0+1} - p_{m,i}| 
			\asymp m\ep|\ell_{i_0+1} - \ell_i| 
	\end{equation}
	and similarly for $i>i_0+1$, we have 
	\begin{equation}
		m\ep|\ell_{i_0+1} - \ell_i| 
			\lesssim |k - p_{m,i}| 
			\lesssim m\ep|\ell_{i_0} - \ell_i|. 
	\end{equation}
	We notice that in both cases, the lower bound is exactly $\frac{m\ep}{\sqrt{\nu_n(i)}}$. The upper bounds, in both cases are 
	\begin{equation}
		m\ep \left|\frac1{\sqrt{\nu_n(1)}} + \frac1{\sqrt{\nu_n(i)}}\right| 
			\leq \frac{2m\ep}{\sqrt{\nu_n(i)}}. 
	\end{equation}
	We hence find $\sin^{-2}kL_i \asymp (m\ep)^{-2} \nu_n(i)$ for all $i\neq i_0,i_0+1$. Thus, by~(\ref{eq: mass measure via G}) and because $\nu_n$ is a probability measure, 
		\[ \mu_n^k(i) 
			\asymp \begin{cases}
				\nu_n(i), & i\neq i_0,i_0+1, \\ 
				\nu_n(1), & i = i_0,i_0+1. 
			\end{cases} \] 
	We now apply Lemma~\ref{lem: stability mass exponent inequality} twice where the finite set is $\{i_0,i_0+1\}$, and we obtain that the mass exponents of $\mu_n^k$ and $\nu_n$ coincide.  
\end{proof}

\printbibliography

\end{document}